\documentclass{llncs}
\title{A Separation of $\gamma$ and $b$ via Thue--Morse Words}
\date{}

\usepackage{amsmath}
\usepackage{url}
\usepackage{hyperref}
\usepackage{graphicx}

\newcommand{\tm}[1]{t_{#1}}

\newcommand{\ground}[1]{\#_g(#1)}
\author{
  Hideo~Bannai\inst{1,4}\orcidID{0000-0002-6856-5185}
\and
  Mitsuru~Funakoshi\inst{2}\orcidID{0000-0002-2547-1509}
\and
  Tomohiro~I\inst{3}\orcidID{0000-0001-9106-6192}
\and
 Dominik~Koeppl\inst{1}\orcidID{0000-0002-8721-4444}
\and
  Takuya~Mieno\inst{2}\orcidID{0000-0003-2922-9434}
\and
  Takaaki~Nishimoto\inst{4}
}
\institute{
  M\&D Data Science Center, Tokyo Medical and Dental University, Japan\\
  \email{\{hdbn,koeppl\}.dsc@tmd.ac.jp}
\and
  Department of Informatics, Kyushu University, Japan\\
  \email{\{mitsuru.funakoshi,takuya.mieno\}@inf.kyushu-u.ac.jp}
\and
  Department of Artificial Intelligence, Kyushu Institute of Technology, Japan\\
  \email{tomohiro@ai.kyutech.ac.jp}
\and
  RIKEN Center for Advanced Intelligence Project, Japan\\
  \email{takaaki.nishimoto@riken.jp}
}

\begin{document}
\maketitle
\begin{abstract}
  We prove that for $n\geq 2$, the size $b(t_n)$ of the smallest bidirectional scheme for the $n$th Thue--Morse word $t_n$ is $n+2$.
  Since Kutsukake et al. [SPIRE 2020] show that the size $\gamma(t_n)$ of the smallest string attractor for $t_n$ is $4$
  for $n \geq 4$, this shows for the first time that
  there is a separation between the size of the smallest string attractor $\gamma$ and the size of the smallest bidirectional scheme $b$, i.e., there exist string families such that $\gamma = o(b)$.
\end{abstract}
\pagestyle{plain}

\section{Introduction}
Repetitiveness measures for strings is an important topic in the field of string compression and indexing.
Compared to traditional entropy-based measures, measures based on dictionary compression
are known to better capture the repetitiveness
in highly repetitive string collections~\cite{10.1145/3434399}.
Some well known examples of dictionary-compression-based measures are:
the size $r$ of the run-length Burrows--Wheeler transform~\cite{Burrows94ablock-sorting} (RLBWT),
the size $z$ of the Lempel-Ziv 77 factorization~\cite{DBLP:journals/tit/ZivL77},
the size $b$ of the smallest bidirectional (or macro) scheme~\cite{DBLP:journals/jacm/StorerS82}.

Kempa and Prezza introduced the notion of {\em string attractors}~\cite{DBLP:conf/stoc/KempaP18},
which gave a unifying view of dictionary-compression-based measures.
A string attractor of a string is a set of positions such that any substring of the string has at least one occurrence which contains a position in the set.
The size $\gamma$ of the smallest string attractor of a word is a lower bound on the size
of all known dictionary compression measures, but is NP-hard to compute.
Kociumaka et al.~\cite{DBLP:conf/latin/KociumakaNP20,kociumaka2021definitive} introduced
another measure $\delta \leq \gamma$ that is computable in linear time,
defined as the maximum over all integers $k$, the number of distinct substrings of
length $k$ in the string divided by $k$.

The landscape of the relations between these measures has been a focus of attention.
For example, since $z$ is a special case of a bidirectional scheme, $b \leq z$.
Also, $b \leq 2r$~\cite{DBLP:journals/tit/NavarroOP21} and
$r = O(z\log^2 N)$~\cite{DBLP:conf/focs/KempaK20} hold, where $N$ is the length of the string.
Notice that a string can be represented in space (with an extra factor of $\log N$ for bits)
proportional to $b$, $r$, or $z$.
Interestingly, while $\delta$ and $\gamma$ do not give a direct representation of the string,
it is known that the string can be represented
in $O(\delta\log\frac{N}{\delta})$ or $O(\gamma\log\frac{N}{\gamma})$ space, respectively~\cite{DBLP:conf/stoc/KempaP18,DBLP:conf/latin/KociumakaNP20,kociumaka2021definitive}.
On the other hand, Kociumaka et al.~\cite{DBLP:conf/latin/KociumakaNP20,kociumaka2021definitive}
showed that for every length $N$ and integer $\delta \in [2,N]$,
there exists a family of length-$N$ strings having the same measure $\delta$,
that requires $\Omega(\delta\log\frac{N}{\delta} \log N)$ bits to be encoded.
Analogous results for $\gamma$ are not yet known~\cite{DBLP:conf/latin/KociumakaNP20,kociumaka2021definitive,10.1145/3434399}.
The bidirectional scheme is the most powerful among the dictionary-compression-based measures.
The size $b$ of the smallest bidirectional scheme is also known
to satisfy $b= O(\gamma\log \frac{N}{\gamma})$, but again, the tightness of this bound was not known~\cite{10.1145/3434399}.

Following Mantaci et al.~\cite{DBLP:conf/ictcs/MantaciRRRS19,DBLP:journals/tcs/MantaciRRRS21},
Kutsukake et al.~\cite{DBLP:conf/spire/KutsukakeMNIBT20} investigated
repetitiveness measures on Thue--Morse words~\cite{prouhet1851,thue1906,morse1921}
and showed that the size of the smallest string attractor for the $n$-th Thue--Morse word is $4$,
for any $n\geq 4$. They also conjectured that the size of the smallest bidirectional scheme
for the $n$-th Thue--Morse word (which has length $N=2^n$) is $\Theta(\log N)$,
which would imply a separation between $\gamma$ and $b$.
Possibly due to the difficulty (NP-hardness)
of computing the size of the smallest bidirectional scheme of a string~\cite{DBLP:journals/jacm/StorerS82},
tight bounds for $b$ have only been discovered for a very limited family of strings, most notably
standard Sturmian words~\cite{DBLP:journals/ipl/MantaciRS03}.
This was shown from the fact that the size $r$ of the RLBWT of every standard Sturmian word is $2$,
therefore implying a constant upper bound on the smallest bidirectional scheme.

In this paper, we prove Kutsukake et al.'s conjecture
by showing that for any $n\geq 2$, the size $b(t_n)$ of the smallest bidirectional scheme for $t_n$ is exactly $n+2$.
For any value of $\gamma \geq 4$, we can construct a family of strings such that
$b = \Theta(\gamma\log\frac{N}{\gamma})$ and $N$ is the length of the string.
Our result shows for the first time the separation between $\gamma$ and $b$,
i.e., there are string families such that $\gamma=o(b)$.

\section{Preliminaries}
We consider the alphabet $\Sigma = \{ \mathtt{a}, \mathtt{b}\}$.
A string is an element of $\Sigma^*$.
For any string $w \in \Sigma^*$,  let $|w|$ denote its length, and
let $w = w[0]\cdots w[|w|-1]$.
Also, for any $0 \leq i \leq j < |w|$,
let $w[i..j] = w[i]\cdots w[j]$.

A \emph{string morphism} $\mu$ is a function mapping strings to strings such that each character is replaced by a single string (deterministically), i.e., $\mu(w) = \mu(w[0])\cdots\mu(w[|w|-1])$ for any string $w$.
Let $\mu^0(w) = w$, and for any integer $n \geq 1$, let $\mu^n(w) = \mu(\mu^{n-1}(w))$.
Now let $\mu$ be the morphism on the binary alphabet determined by $\mu(\mathtt{a}) = \mathtt{ab}$ and $\mu(\mathtt{b}) = \mathtt{ba}$.
Then the $n$-th \emph{Thue--Morse word} $t_n$ is $\mu^n(\mathtt{a})$, and its length is $|t_n| = 2^n$.

A list of strings $b_1,\ldots,b_k$ is called a {\em parsing} of a string $S$,
if $S = b_1\cdots b_k$. Each $b_i~(i=1,\ldots,k)$ is called a {\em phrase}.
A sequence $B= ((b_1,s_1), \ldots, (b_k,s_k))$
is a {\em bidirectional scheme} for $S$,
if $b_1,\ldots, b_k$, is a parsing of $S$
and for all $i = 1,\ldots, k$,
$s_i \in [0,|S|-1]\cup\{\bot\}$, such that
$s_i = \bot$ if $|b_i|=1$, and $b_i =S[s_i..s_i+|b_i|-1]$ otherwise.
We denote the {\em size} $k$ of the bidirectional scheme~$B$ by $|B|$.
We call $s_i$ the {\em source} of the phrase~$b_i$.

If $|b_i| = 1$ then we stipulate that $s_i=\bot$, and call $b_i$ a {\em ground} phrase.
(Consequently, there are no phrases of length one that have a source being a text position.)
We denote the number of ground phrases in $B$ by $\ground{B}$.
For convenience, we denote the starting position of phrase $b_i$ by $p_i$,
i.e., $p_1=0$ and $p_i=|b_1\cdots b_{i-1}|$ for all $i=2,\ldots,k+1$.

A bidirectional scheme $B$ for the string $S$ defines a
function $f_B:[0,|S|-1]\cup\{\bot\} \rightarrow [0,|S|-1]\cup\{\bot\}$
over positions of $S$, where
\[
  f_B(x) =
  \begin{cases}
    \bot        & \mbox{if } x = \bot \mbox { or if } x = p_i, s_i =\bot                    \mbox{ for some $i$,} \\
    s_i+x-p_{i} & \mbox{otherwise, i.e., if }p_{i}\leq x < p_{i+1}, s_i\neq\bot \mbox{ for some $i$.}             \\
  \end{cases}
\]
Let $f_B^0(x) = x$, and for any $j\geq 1$, let $f_B^j(x) = f_B(f_B^{j-1}(x))$.
It is clear that if $f_B(i) \neq \bot$, then it holds that $S[i] = S[f_B(i)]$.
A bidirectional scheme for $S$ is {\em valid}, if
there is no $i\in [0,|S|-1]$ such that
the function~$f_B$ contains a cycle, that is,
for every $i \in [0,|S|-1]$, there exists a $j\geq 1$ such that $f_B^j(i) = \bot$.
A valid bidirectional scheme $B$ of size $k$ for $S$ implies an $O(k)$-word size (compressed) representation of $S$,
namely, the sequence $((|b_1|,s'_1), \ldots, (|b_k|,s'_k))\subset ([1,|S|]\times\{[0,|S|-1]\cup\Sigma\})^k$,
where $s'_i = b_i$ if $s_i=\bot$, and $s'_i=s_i$ otherwise.
Note that the string $S$ can be reconstructed from this sequence if and only if $B$ is valid.
A parsing $b_1, \ldots, b_k$ of~$S$ is \emph{valid} if there exists
a list of phrase sources $s_1,\dots, s_k$ such that $((b_1,s_1), \ldots, (b_k,s_k))$ is a valid bidirectional scheme for $S$.

Informally, $f_B(x)$ gives the position (source) from where we want to copy the character that restores $S[x]$ when reconstructing
$S$ from the compressed representation,
where $f_B(x) = \bot$ indicates that the character is stored as a ground phrase, i.e., as a literal.

It is easy to see that a valid bidirectional scheme must have at least as many ground phrases
as there are different characters appearing in~$S$ (the number of ground phrases is at least $|\Sigma|$ if all characters of $\Sigma$ appear in $S$).

\section{Important Characteristics of Thue--Morse Words}
Before proving our bounds, we first give some simple observations on Thue--Morse words that we will use later.
Remember that the first index of $t_n$ is 0, which is an even position.

\begin{lemma}\label{lem:aabb}
  $\mathtt{aa}$ and $\mathtt{bb}$ only occur at odd positions in $t_n$.
\end{lemma}
\begin{proof}
  The morphism $\mu$ implies that any substring of length 2 starting at an even position is either $\mu(\mathtt{a}) = \mathtt{ab}$ or $\mu(\mathtt{b}) = \mathtt{ba}$.
  \qed
\end{proof}

\begin{lemma}[Theorem 2.2.3 of~\cite{berstel_reutenauer_1997}]\label{lem:overlapfree}
  $t_n$ has no overlapping factors, i.e.,
  two occurrences of the same string in $t_n$ never share a common position.
\end{lemma}

\begin{lemma}\label{lem:abab}
  $\mathtt{abab}$ and $\mathtt{baba}$ only occur at even positions in $t_n$.
\end{lemma}
\begin{proof}
  Suppose to the contrary that there is an occurrence of $\mathtt{abab}$ that starts at an odd position.
  Then, Lemma~\ref{lem:aabb} implies that $\mathtt{b}$ occurs immediately left of $\mathtt{abab}$,
  i.e., there is an occurrence of the substring $\mathtt{babab}$, thus contradicting Lemma~\ref{lem:overlapfree}
  with the substring $\mathtt{bab}$ having two overlapping occurrences.
  \qed
\end{proof}

Let the {\em parity} of an integer $i$ be $i \bmod 2 \in \{0, 1\}$.
\begin{lemma}\label{lem:parity}
  For any substring $w \not \in \{\mathtt{aba}, \mathtt{bab}, \mathtt{ab}, \mathtt{ba}, \mathtt{a}, \mathtt{b}\}$ of $t_n$,
  the parities of all occurrences of $w$ in $t_n$ are the same.
\end{lemma}
\begin{proof}
  Any such substring $w$ contains at least one of $\{\mathtt{aa}, \mathtt{bb}, \mathtt{abab}, \mathtt{baba}\}$
  as a substring, and thus the result follows from Lemmas~\ref{lem:aabb} and~\ref{lem:abab}.\qed
\end{proof}

Further, we use that $t_n$ is a prefix of $t_{n+1}$ and $t_n[0..4] = \mathtt{abbab}$ for $n \ge 3$.

\section{Upper and Lower Bounds on $b$}

We start with the upper bound on the smallest size of a (valid) bidirectional parsing by constructing such a parsing,
and subsequently show that this bound is optimal by showing a lower bound whose proof is more involved.
\subsection{Upper Bound}
\begin{theorem}[Upper bound]\label{thm:upperbound}
  For $n\geq 2$,
  there exists a valid bidirectional scheme for $\tm{n}$ of size $n+2$.
\end{theorem}
\begin{proof}
  Proof by induction.
  For $n = 2$ it is clear that there is a valid bidirectional scheme of size $4$.

  Suppose that for some $n\geq 2$, there is a valid bidirectional scheme
  $B_n = ((b_1,s_1),\ldots,(b_{k},s_{k}))$ of size $k$ for $\tm{n}$.
  We can assume that there are at least two ground phrases
  $b_{i_\mathtt{a}} = \tm{n}[p_{i_\mathtt{a}}]=\mathtt{a}$ and
  $b_{i_\mathtt{b}} = \tm{n}[p_{i_\mathtt{b}}]=\mathtt{b}$.
  Since $\tm{n+1}=\mu(\tm{n})$, we first consider a bidirectional scheme
  $B'$ for $\tm{n+1}$ where each phrase is constructed from phrases of $B_n$
  by applying $\mu$, with the small exception for the two ground phrases.
  More precisely, the phrases of $B'$ are
  $\mu(b_i)$ for $i \in [1,k]\setminus\{i_{\mathtt{a}},i_{\mathtt{b}}\}$,
  and two ground phrases from each of
  $\mu(b_{i_\mathtt{a}})=\mathtt{ab}$ and
  $\mu(b_{i_\mathtt{b}})=\mathtt{ba}$,
  resulting in a parsing of size $k+2$.
  For each non-ground phrase $\mu(b_i)$ in $B'$,
  we can either choose the source to be
  (i) $2p_{i_\mathtt{a}}$ or $2p_{i_\mathtt{b}}$
  if its length is $2$, or (ii) $2s_i$ otherwise.
  The latter is because
  $\mu(b_i) =\mu(\tm{n}[s_i..s_i+|b_i|-1]) = \mu(\tm{n})[2s_i..2s_i+2|b_i|-1] = \tm{n+1}[2s_i..2s_i+2|b_i|-1]$.
  The validity of $B'$ follows from the validity of $B_{n}$,
  and $f_{B'}$ has no cycles.
  It is easy to see that
  for any position $i$, the parities of $i$ and $f_{B'}(i)$ are the same (unless $f_{B'}(i)=\bot$).
  Thus,
  noticing that $\tm{n+1}[3..4] = \mathtt{ab}$,
  (1) the source of $t_{n+1}[3]=\mathtt{a}$ at an odd position can eventually be traced to the ground phrase at position $2p_{i_\mathtt{b}}+1$,
  and (2) the source of $t_{n+1}[4]=\mathtt{b}$ at an even position can eventually be traced to the ground phrase at position $2p_{i_\mathtt{b}}$.

  Next,
  we modify $B'$
  by combining the two consecutive ground phrases $(\mathtt{a},\bot)$ and $(\mathtt{b},\bot)$
  corresponding to $\mu(b_{i_\mathtt{a}})$,
  and replace them with a single $(\mathtt{ab},3)$.
  This results in a bidirectional scheme $B''$ of size $k+1$.
  From the above observations (1) and (2), it is clear that $B''$ is still valid.
  Thus, $B_{n+1} = B''$ is a valid bidirectional scheme for $t_{n+1}$ of size $k+1$, thereby proving the theorem.
  \qed
\end{proof}

\subsection{Lower Bound}
\begin{theorem}[Lower Bound]\label{thm:lowerbound}
  For $n\geq 2$, the smallest valid bidirectional scheme for $\tm{n}$ has size $n+2$.
\end{theorem}

To prove Theorem~\ref{thm:lowerbound}, we would like to, in essence, do the opposite of what we did in the proof of Theorem~\ref{thm:upperbound},
and show that we can construct a bidirectional scheme for $t_{n-1}$ of size $k-1$, given a bidirectional scheme for $t_n$ of size $k$.
However, the opposite direction involves halving the size of phrases, and thus does not work straightforwardly.
Nevertheless, we will show that this can be done in an amortized way,
and show the following.

\begin{lemma}\label{lemma:main}
  For any $n \geq 5$, if there exists a valid bidirectional scheme of size $k$ for $t_{n}$,
  then, for some $1 \leq i \leq 3$, there exists a valid bidirectional scheme of size at most $k-i$ for $t_{n-i}$.
\end{lemma}

Since the size of the smallest bidirectional scheme for $t_2$, $t_3$, $t_4$
can be confirmed to be respectively $4, 5, 6$ by computer analysis,
this with Lemma~\ref{lemma:main} implies Theorem~\ref{thm:lowerbound}.

In the rest of the section, we give an algorithm that, given a bidirectional scheme $B_n$ for $t_n$,
constructs a bidirectional scheme $B_{n-1}$ for $t_{n-1}$,
and claim that applying the algorithm repeatedly
$i$ times, for some $1 \leq i \leq 3$, we obtain
a bidirectional scheme $B_{n-i}$ for $t_{n-i}$ such that $|B_{n-i}| \leq |B_{n}|-i$.
The algorithm consists of 3 main steps:
\begin{enumerate}
  \item\label{step:elim_ground} Elimination of length-$1$ ground phrases.
  \item\label{step:elim_odd} Elimination of odd length phrases.
  \item\label{step:inverse_morphism} Application of the inverse morphism $\mu^{-1}$ on all phrases of the modified parsing.
\end{enumerate}

The goal of Steps~\ref{step:elim_ground} and~\ref{step:elim_odd} is to modify the phrases of $B_{n}$
to construct a bidirectional scheme $B'_n$ so that all phrases in $B'_n$ will be of even length.
When modifying the phrases, we must take care in 1) defining the source of the phrase,
and 2) ensuring that no cycles are introduced in the resulting bidirectional scheme $B_{n-1}$.
To make this clear, we temporarily relax the definition for ground phrases in $B'_n$ during the modification,
so that the ground phrases of $B'_n$ are phrases of length $2$ that start at even positions.
In this way, we can be sure that any position in a length-$2$ phrase starting at an even position in $B'_n$ is not involved in a cycle.
In Step~\ref{step:inverse_morphism}, we create a new bidirectional scheme~$B_{n-1}$ of $t_{n-1}$ by translating all phrase lengths and sources of $B'_n$ according to the inverse morphism $\mu^{-1}$,
i.e., we map each non-ground phrase $(b'_i,s'_i)$ of $B'_n$ to
the phrase $(\mu^{-1}(b'_i),s'_i/2)$ in $t_{n-1}$.
The length-$2$ ground phrases in $B'_n$ become length-$1$ ground phrases in $B_{n-1}$, and thus
we obtain a valid bidirectional scheme $B_{n-1}$ for $t_{n-1}$, without the relaxation, and the same size as $B'_n$.

\subsubsection{Eliminating Length-$1$ Ground Phrases}\label{secLengthOne}
The operation is done analogously and symmetrically for
any length-$1$ ground phrase ($\mathtt{a}$ or $\mathtt{b}$) that may occur at an even or odd position.
We describe in detail the case for a ground phrase with character $\mathtt{a}$ that occurs at some odd position $2i+1$.

For a consecutive pair of positions $2i,2i+1$, we call one a {\em partner} of the other.
Let $i_\mathtt{b} = 2i$ be the partner position of the length-$1$ ground phrase $\mathtt{a}$, i.e.,
$t_n[i_\mathtt{b}..i_\mathtt{b}+1] = \mathtt{ba}$.
The idea is to (re)move the phrase boundary that separates partner positions so that the ground phrase disappears.
Since we are considering the case where the ground phrase is at an odd position,
we extend the phrase $(b_i, s_i)$ containing position $i_\mathtt{b}$ by one character,
so that it includes the length-$1$ ground phrase $t_n[i_\mathtt{b}+1]=\mathtt{a}$, thereby eliminating it.
If possible, we would like to keep the source of the extended phrase the same, i.e.,
change $(b_i, s_i)$ to $(b_i\mathtt{a},s_i)$,
or equivalently, change $f_{B_n}(i_\mathtt{b}+1) =\bot$ to $f_{B_n}(i_\mathtt{b}+1) = s_i+|b_i|$.
Note that if the parity of $f_{B_n}(i_\mathtt{b})$ is equal to that of
$i_\mathtt{b}$,
this is always possible (i.e., $t_n[f_{B_n}(i_\mathtt{b})+1]=\mathtt{a}$ always holds).
However, it may be that the position $i_\mathtt{b}+1$ gets involved in a cycle, due to this change.
Notice that since we started from a valid (relaxed) bidirectional scheme, it is guaranteed that
$i_\mathtt{b}$ is not involved in a cycle, i.e.,
$f^j_{B_n}(i_\mathtt{b}) \neq i_\mathtt{b}$ for any $j \geq 1$.
Therefore, we further modify the phrase boundaries, if necessary,
to ensure that the source of $t_n[i_\mathtt{b}+1]=\mathtt{a}$ will belong in the same phrase as the source of $t_n[i_\mathtt{b}]=\mathtt{b}$.
This is repeated until we are sure that all these changes made to eliminate the original length-$1$ ground phrase $\mathtt{a}$ do not introduce any cycles in the final bidirectional scheme.
In other words,  we ensure, for some {\em sufficiently large} $j'$,
$f^{j}_{B_n}(i_\mathtt{b}+1) = f^{j}_{B_n}(i_\mathtt{b})+1$ for all $1\leq j \leq j'$.
Then, from the acyclicity of position $i_\mathtt{b}$, the acyclicity of position $i_\mathtt{b}+1$ follows.

There are six cases where the process terminates, as shown in Figure~\ref{fig:elim_ground} (Case 3 is further divided into two sub-cases).
As noted above, as long as the parity of $f^j_{B_n}(i_\mathtt{b})$ is the same as
that of $f^{j-1}_{B_n}(i_\mathtt{b})$,
the character of $f^j_{B_n}(i_\mathtt{b})$'s partner is always $\mathtt{a}$, and we can ensure that
$f^{j-1}_{B_n}(i_\mathtt{b})$ and $f^{j-1}_{B_n}(i_\mathtt{b}+1)$
are in the same phrase by only (possibly) setting
$f^j_{B_n}(i_\mathtt{b}+1) = f^j_{B_n}(i_\mathtt{b})+1$.
Thus, we consider the cases where
$j'\geq 1$ is the smallest integer such that
the parities of $f^{j'-1}_{B_n}(i_\mathtt{b})$ and $f^{j'}_{B_n}(i_\mathtt{b})$ differ,
in which case, Lemma~\ref{lem:parity} implies that
$f^{{j'}-1}_{B_n}(i_\mathtt{b})$ is contained in a phrase in
$\{\mathtt{aba}, \mathtt{bab}, \mathtt{ab}, \mathtt{ba}, \mathtt{b}  \}$.
Each of the six cases corresponds to a distinct occurrence of $\mathtt{b}$ in the strings of this set.
We show that in each case, we can modify the phrases so that both
$f^{j'-1}_{B_n}(i_\mathtt{b})$ and $f^{j'-1}_{B_n}(i_\mathtt{b}+1)$
are in the same length-2 phrase, i.e., a relaxed ground phrase, and be sure
that $i_\mathtt{b}+1$ will not be involved in a cycle in the final bidirectional scheme.
The details of each case are described in Figure~\ref{fig:elim_ground}.

Although Cases 1, 2, 4 introduce a new length-$1$ ground phrase,
the number of phrase boundaries that separate partner positions always decreases at the starting point, and never increases.
Therefore the whole process terminates at some point, at which point, all length-$1$ ground phrases have been eliminated.

\begin{figure}[h!]
  \centerline{
    \begin{tabular}{c|c}\hline
      \parbox[c]{0.24\textwidth}{
        \includegraphics[width=0.23\textwidth]{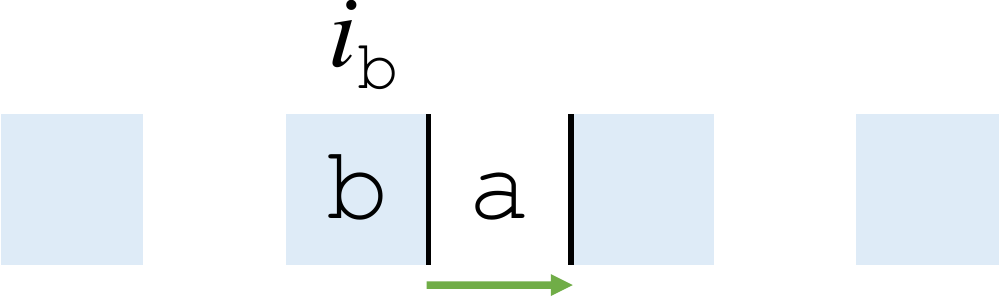}
      }
       &
      \parbox[c][][t]{0.75\textwidth}{ \vspace{2pt}
      Starting point.
      We eliminate the length-$1$ ground phrase $t_n[i_\mathtt{b}+1]=\mathtt{a}$ at an odd position $i_\mathtt{b}+1$,
      by including it in the same phrase as is partner $t_n[i_\mathtt{b}]=\mathtt{b}$, in this case, on its left.
      We can do this by modifying the source of $\mathtt{a}$ to point to the position next to the source of $\mathtt{b}$, i.e.,
      setting $f_{B_n}(i_\mathtt{b}+1) = f_{B_n}(i_\mathtt{b})+1$.
      This is done recursively at the source positions, until we reach one of the following cases,
      where the
      source of $\mathtt{b}$ no longer points to a position of the same parity.
      }  \\\hline
      \parbox[c]{0.24\textwidth}{
        \includegraphics[width=0.23\textwidth]{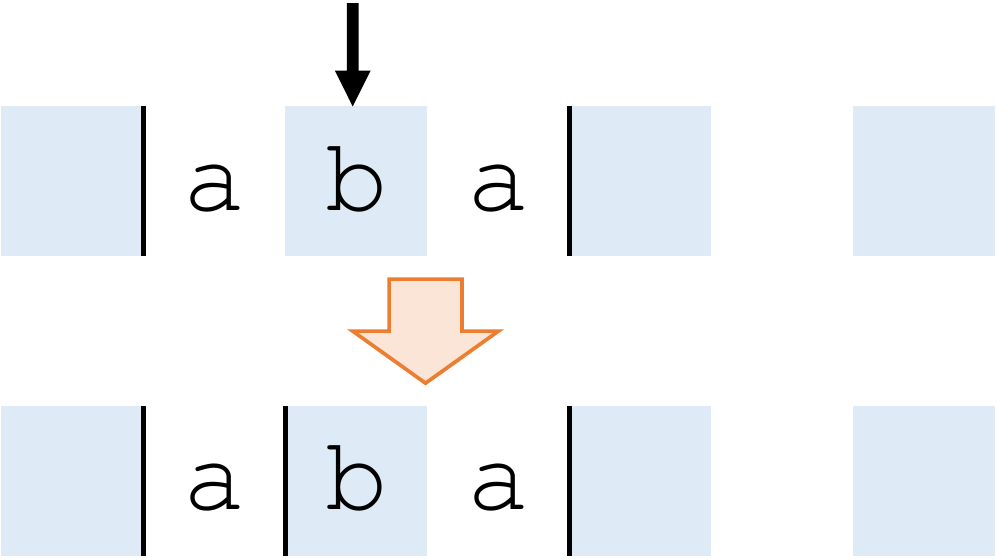}
      }
       &
      \parbox[c]{0.75\textwidth}{ \vspace{2pt}
      Case 1:
      We introduce a new length-$1$ ground phrase, and modify the boundaries.
      We are done since both
      $t_n[f_{B_n}^{j'-1}(i_\mathtt{b})]=\mathtt{b}$ and $t_n[f_{B_n}^{j'-1}(i_\mathtt{b})+1]=\mathtt{a}$
      are in a length-$2$ phrase starting at an even position, i.e., a relaxed ground phrase.
      We are sure that $i_\mathtt{b}+1$ is not involved in a cycle.
      We recursively apply the procedure to the new length-$1$ ground phrase $\mathtt{a}$ at an odd position.
      }  \\\hline
      \parbox[c]{0.24\textwidth}{
        \includegraphics[width=0.23\textwidth]{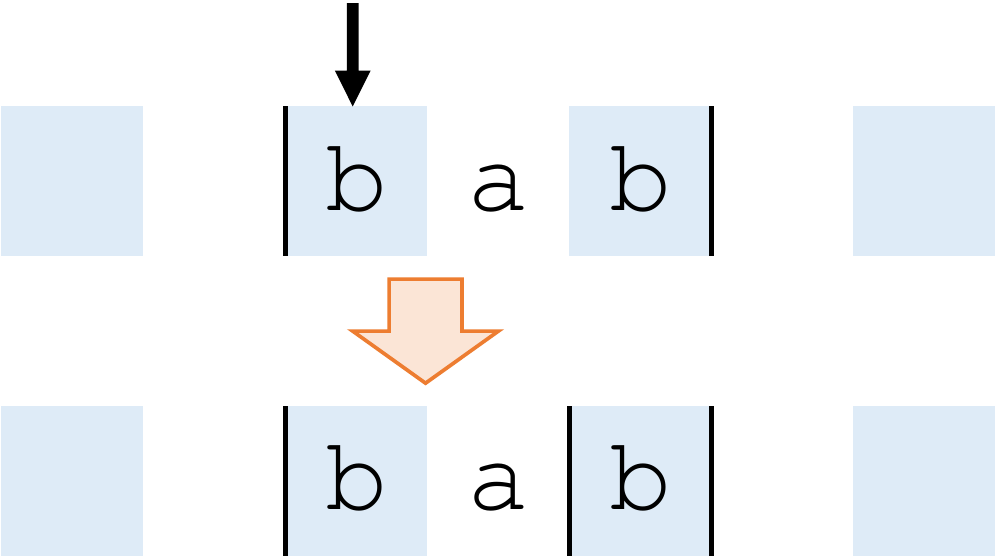}
      }
       &
      \parbox[c]{0.75\textwidth}{
        Case 2:
        Same as Case 1, with the exception that the new length-1 ground phrase is $\mathtt{b}$ at an even position.
      }  \\\hline
      \parbox[c]{0.24\textwidth}{
        \includegraphics[width=0.23\textwidth]{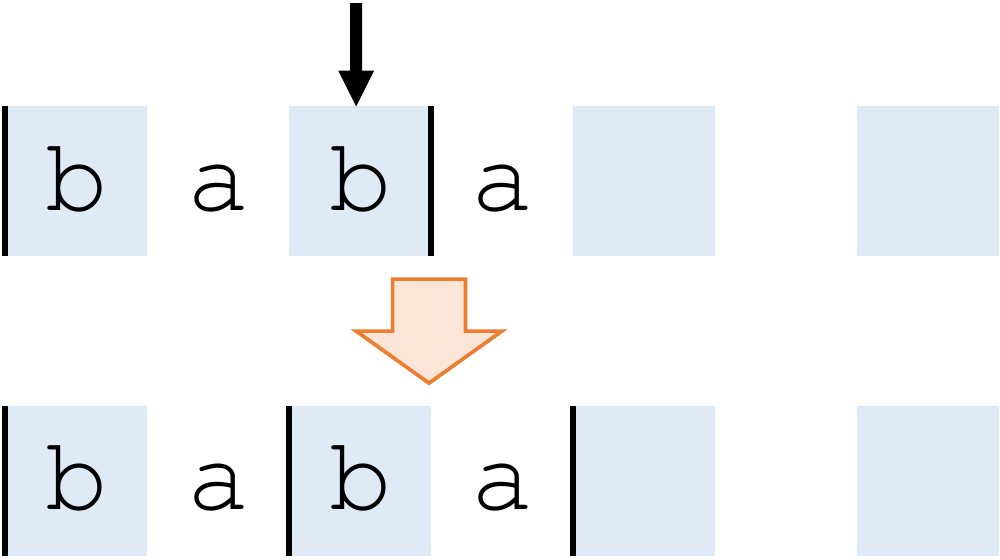}
      }
       &
      \parbox[c]{0.75\textwidth}{ \vspace{2pt}
      Case 3-1:
      This case is when there are no consecutive phrases of $\mathtt{ba}$ and $\mathtt{ba}$ in the current bidirectional scheme.
      We introduce a new length-$2$ phrase, and modify the boundaries.
      We are done since both
      $t_n[f_{B_n}^{j'-1}(i_\mathtt{b})]=\mathtt{b}$ and $t_n[f_{B_n}^{j'-1}(i_\mathtt{b})+1]=\mathtt{a}$
      are in a length-$2$ phrase starting at an even position, i.e., a relaxed ground phrase.
      }  \\\hline
      \parbox[c]{0.24\textwidth}{
        \includegraphics[width=0.23\textwidth]{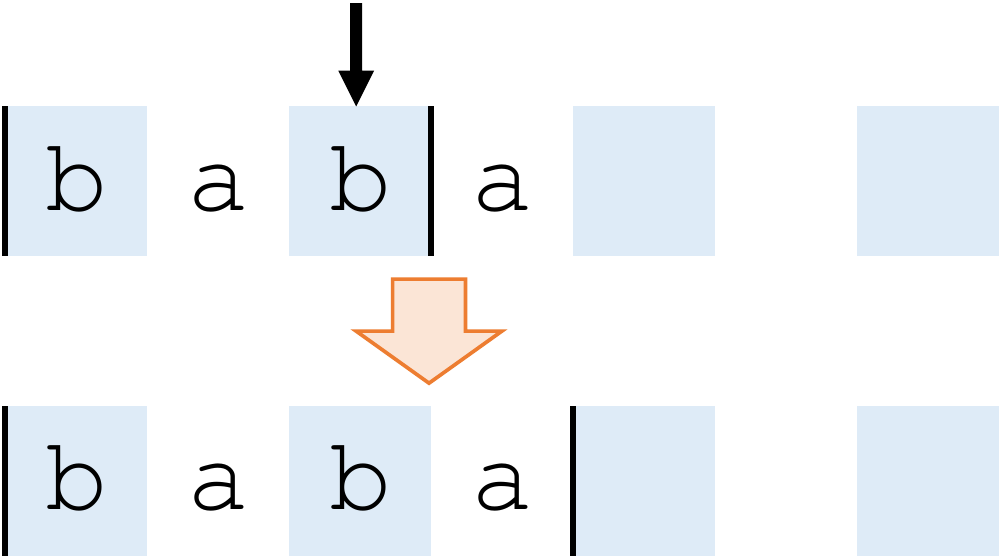}
      }
       &
      \parbox[c]{0.75\textwidth}{ \vspace{2pt}
        Case 3-2:
        This case is when there already are
        consecutive phrases of $\mathtt{ba}$ and $\mathtt{ba}$ in the current bidirectional scheme.
        We create the phrase $\mathtt{baba}$ and make its source be the consecutive phrases of $\mathtt{ba}$ and $\mathtt{ba}$
        (possibly constructed in Case 3-1). No cycles are introduced since the new source are relaxed ground phrases.
      }  \\\hline
      \parbox[c]{0.24\textwidth}{
        \includegraphics[width=0.23\textwidth]{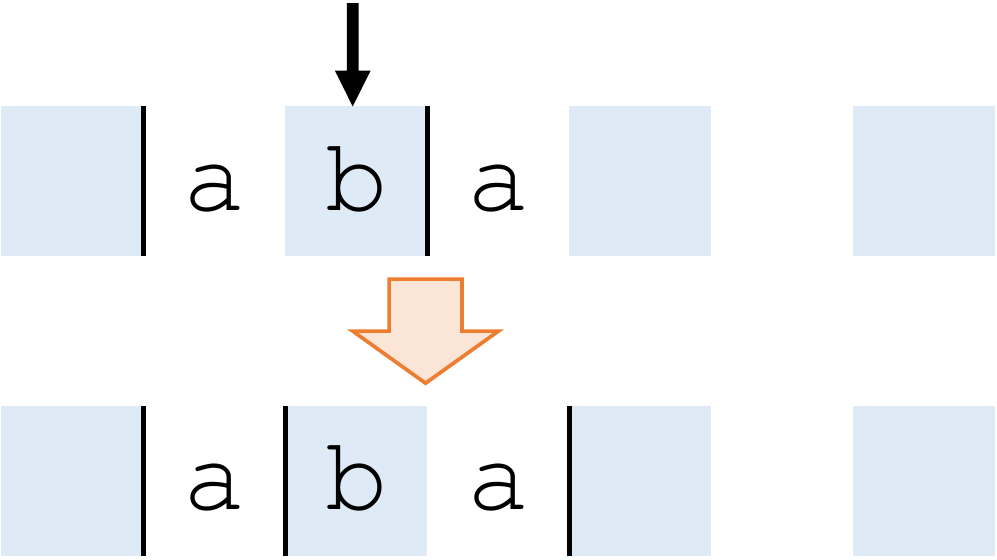}
      }
       &
      \parbox[c]{0.75\textwidth}{
        Case 4:
        Same as Case 1.
      }  \\\hline
      \parbox[c]{0.24\textwidth}{
        \includegraphics[width=0.23\textwidth]{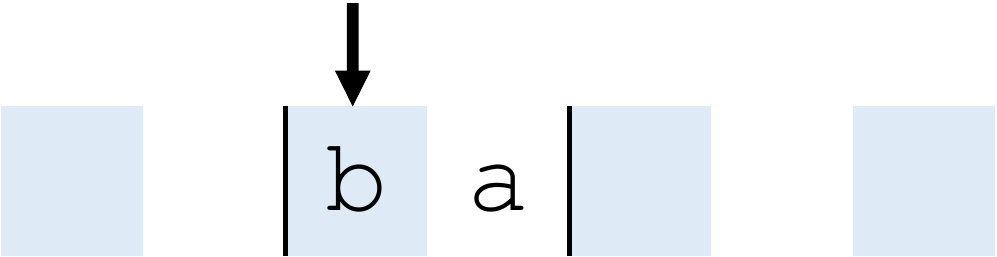}
      }
       &
      \parbox[c]{0.75\textwidth}{ \vspace{2pt}
      Case 5:
      There is nothing to do.
      We are done since both
      $t_n[f_{B_n}^{j'-1}(i_\mathtt{b})]=\mathtt{b}$ and $t_n[f_{B_n}^{j'-1}(i_\mathtt{b})+1]=\mathtt{a}$
      are in a length-$2$ phrase starting at an even position, i.e., a relaxed ground phrase.
      }  \\\hline
      \parbox[c]{0.24\textwidth}{
        \includegraphics[width=0.23\textwidth]{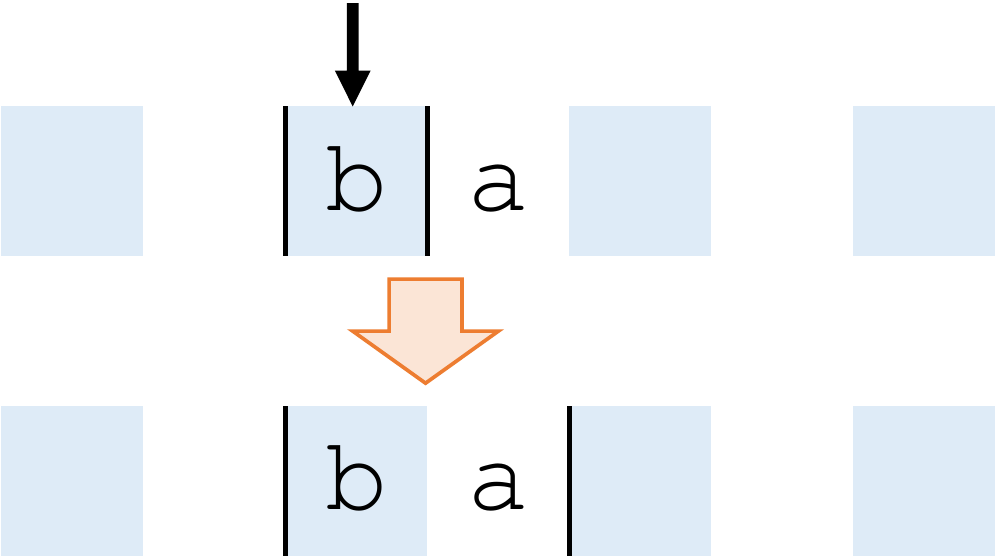}
      }
       &
      \parbox[c]{0.75\textwidth}{ \vspace{2pt}
      Case 6:
      We expand the phrase to include its partner.
      We are done, since both
      $t_n[f_{B_n}^{j'-1}(i_\mathtt{b})]=\mathtt{b}$ and $t_n[f_{B_n}^{j'-1}(i_\mathtt{b})+1]=\mathtt{a}$
      are in a length-$2$ phrase starting at an even position, i.e., a relaxed ground phrase.
      }  \\\hline
    \end{tabular}
  }\caption{
  Terminal cases for eliminating a length-$1$ ground phrase $t_n[i_\mathtt{b}+1]=\mathtt{a}$ at an odd position $i_\mathtt{b}+1$ (see Section~\ref{secLengthOne}).
  The shaded squares are even positions. The vertical bars denote phrase boundaries.
  The black arrow points to the position $f^{j'-1}_{B_n}(i_\mathtt{b})$, where $j'\geq 1$ is the smallest integer such that
  the parities of $f^{j'-1}_{B_n}(i_\mathtt{b})$ and $f^{j'}_{B_n}(i_\mathtt{b})$ differ.
  The first line and second line of each case (except Case 5) respectively show the phrase boundaries before and after the modification.
  }\label{fig:elim_ground}
\end{figure}

\subsubsection{Eliminating Odd Length Phrases}\label{secOddLength}
In this step, we eliminate all remaining phrases with odd lengths.
Since there are no more length-$1$ ground phrases, we first focus on removing phrases $\mathtt{aba}$ and $\mathtt{bab}$ of length $3$.
Below, we describe the operation for removing a phrase $\mathtt{aba}$ that starts at an odd position.
The other cases are analogous or symmetric.

Starting with an occurrence of phrase $\mathtt{aba}$ that starts at an odd position $i_\mathtt{b}+1$,
we know that this phrase is preceded by $\mathtt{b}$.
We move the phrase boundary that separates partner positions,
so that the length-$3$ phrase shrinks to a length-$2$ phrase starting at an even position,
i.e., a relaxed ground phrase, in this case, by expanding the phrase to its left.
Since we have changed the source of the $\mathtt{a}$ at position  $i_\mathtt{b}+1$, we ensure that
for some sufficiently large $j'$,
$f^{j}_{B_n}(i_\mathtt{b}+1) = f^{j}_{B_n}(i_\mathtt{b})+1$ for all $1\leq j\leq j'$,
as we did for the elimination of length-$1$ ground phrases,
so that $i_\mathtt{b}+1$ is not involved in a cycle.

There are five cases where the process terminates, as shown in Figure~\ref{fig:elim_length3}.
As noted previously, as long as the parity of $f^j_{B_n}(i_\mathtt{b})$ is the same as
that of $f^{j-1}_{B_n}(i_\mathtt{b})$,
then the character of $f^j_{B_n}(i_\mathtt{b})$'s partner is always $\mathtt{a}$, and we can ensure that
$f^{j-1}_{B_n}(i_\mathtt{b})$ and $f^{j-1}_{B_n}(i_\mathtt{b}+1)$ are in the same phrase by only (possibly) setting
$f^j_{B_n}(i_\mathtt{b}+1) = f^j_{B_n}(i_\mathtt{b})+1$.
Thus, we consider the cases where
$j'\geq 1$ is the smallest integer such that
the parities of $f^{j'-1}_{B_n}(i_\mathtt{b})$ and $f^{j'}_{B_n}(i_\mathtt{b})$ differ,
in which case, Lemma~\ref{lem:parity} and the previous step implies that
$f^{{j'}-1}_{B_n}(i_\mathtt{b})$ is contained in a phrase in
$\{ \mathtt{aba}, \mathtt{bab}, \mathtt{ab}, \mathtt{ba}  \}$.
Each of the five cases corresponds to a distinct occurrence of $\mathtt{b}$ in strings of this set.
The details of each case are described in Figure~\ref{fig:elim_length3}.

After eliminating all phrases $\mathtt{aba}$ and $\mathtt{bab}$ of length $3$,
all remaining  phrases are either of length $2$ or do not belong to the set $\{ \mathtt{aba},\mathtt{bab},\mathtt{ab},\mathtt{ba},\mathtt{a},\mathtt{b}\}$.
Therefore, we can move all phrase boundaries that separate partner positions to the right (or all of them to the left)
and update the sources accordingly without introducing cycles,
since length-$2$ phrases starting at odd positions become relaxed ground phrases,
and the occurrences of each of the other phrases have the same parity due to Lemma~\ref{lem:parity}.
Thus, we now have a valid bidirectional scheme $B'_n$ where all phrases are of even length, and length-$2$ phrases are considered to be relaxed ground phrases.

\begin{figure}[h!]
  \centerline{
    \begin{tabular}{c|c}\hline
      \parbox[c]{0.25\textwidth}{
        \includegraphics[width=0.23\textwidth]{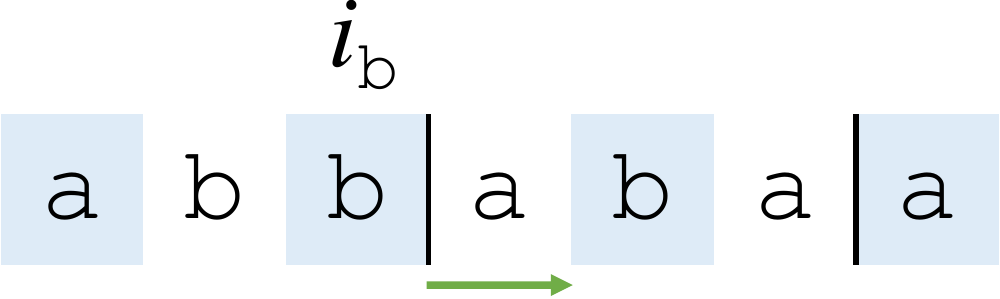}
      }
       &
      \parbox[c]{0.7\textwidth}{ \vspace{2pt}
      Starting point. We wish to eliminate the length-$3$ phrase $\mathtt{aba}$ starting at an odd position $i_\mathtt{b}+1$.
      We move the boundary so that the length-$3$ phrase shrinks to a length-$2$ phrase that starts at an even position.
      In this case, we extend the phrase on its left side to include $t_n[i_\mathtt{b}+1]=\mathtt{a}$.
      We can do this by modifying the source of $\mathtt{a}$ to point to the position next to the source of $\mathtt{b}$, i.e., setting $f(i_\mathtt{b}+1) = f(i_\mathtt{b})+1$.
      This is done recursively at the source positions, until we reach one of the following cases,
      where
      the source of $\mathtt{b}$ no longer points to a position of the same parity.
      }  \\\hline
      \parbox[c]{0.25\textwidth}{
        \includegraphics[width=0.23\textwidth]{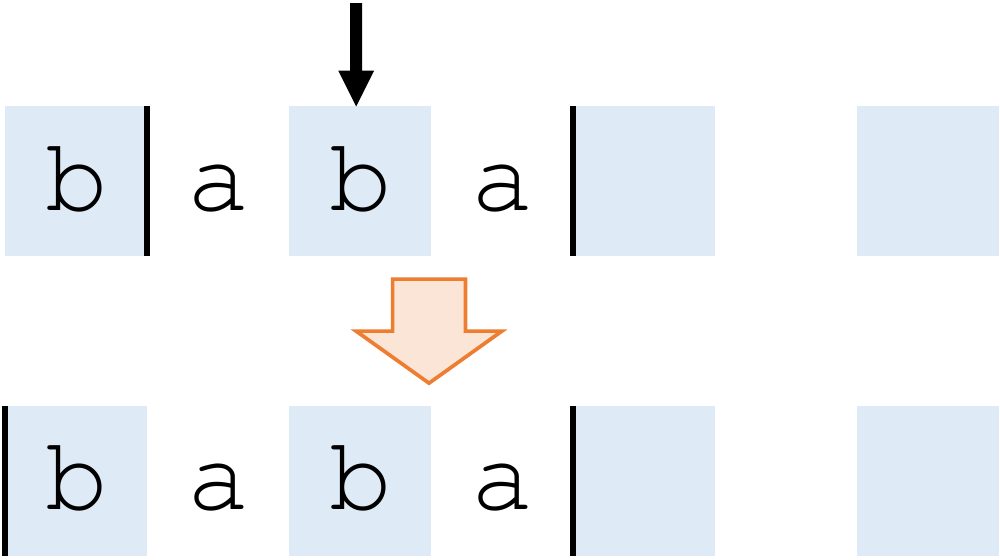}
      }
       &
      \parbox[c]{0.7\textwidth}{ \vspace{2pt}
      Case 1:
      Noticing that the phrase is a substring of $\mathtt{baba}$, we expand the phrase and make the source point to
      $i_\mathtt{b}$. This is possible because $t_n[i_\mathtt{b}..i_\mathtt{b}+3]=\mathtt{baba}$.
      Also, since each of the $\mathtt{ba}$'s can finally be traced to the relaxed ground phrase created at the starting point, i.e., $\mathtt{ba}$ at position $i_\mathtt{b}+2$,
      we are done. \vspace{2pt}
      }  \\\hline
      \parbox[c]{0.25\textwidth}{
        \includegraphics[width=0.23\textwidth]{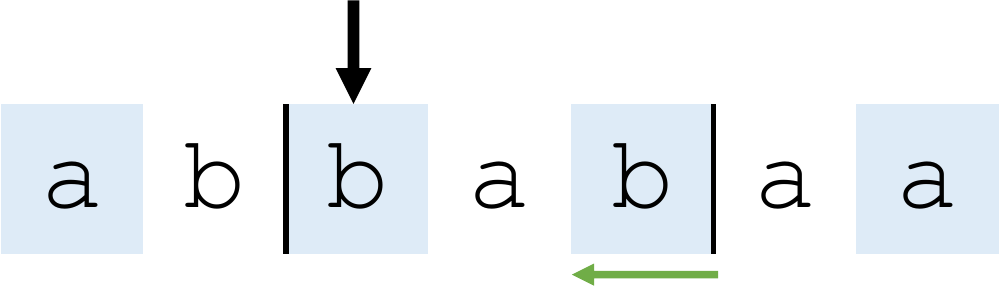}
      }
       &
      \parbox[c]{0.7\textwidth}{ \vspace{2pt}
      Case 2:
      We reach a different length-$3$ phrase $\mathtt{bab}$ that we wish to eliminate, that ends at an even position.
      We recursively apply the procedure to eliminate the phrase $\mathtt{bab}$.
      In this case, that will shrink this length-$3$ phrase to a length-$2$ phrase by expanding
      the phrase to its right.
      Then, we are done with the elimination of the original length-$3$ phrase $\mathtt{aba}$,
      since
      $t_n[f_{B_n}^{j'-1}(i_\mathtt{b})]=\mathtt{b}$ and $t_n[f_{B_n}^{j'-1}(i_\mathtt{b})+1]=\mathtt{a}$
      are in a length-$2$ phrase starting at an even position, i.e., a relaxed ground phrase.
      }  \\\hline
      \parbox[c]{0.25\textwidth}{
        \includegraphics[width=0.23\textwidth]{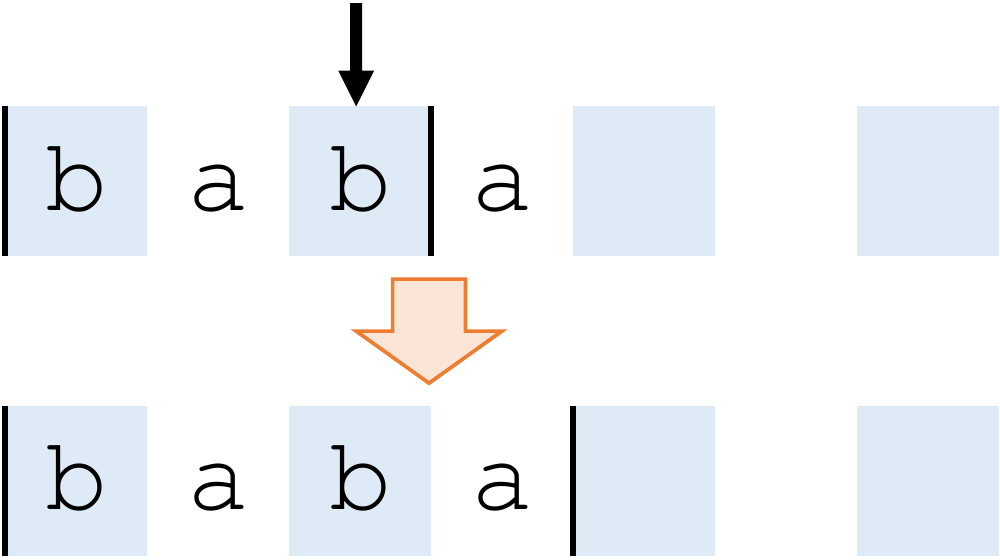}
      }
       &
      \parbox[c]{0.7\textwidth}{
        Case 3:
        Same as Case 1.
      }  \\\hline
      \parbox[c]{0.25\textwidth}{
        \includegraphics[width=0.23\textwidth]{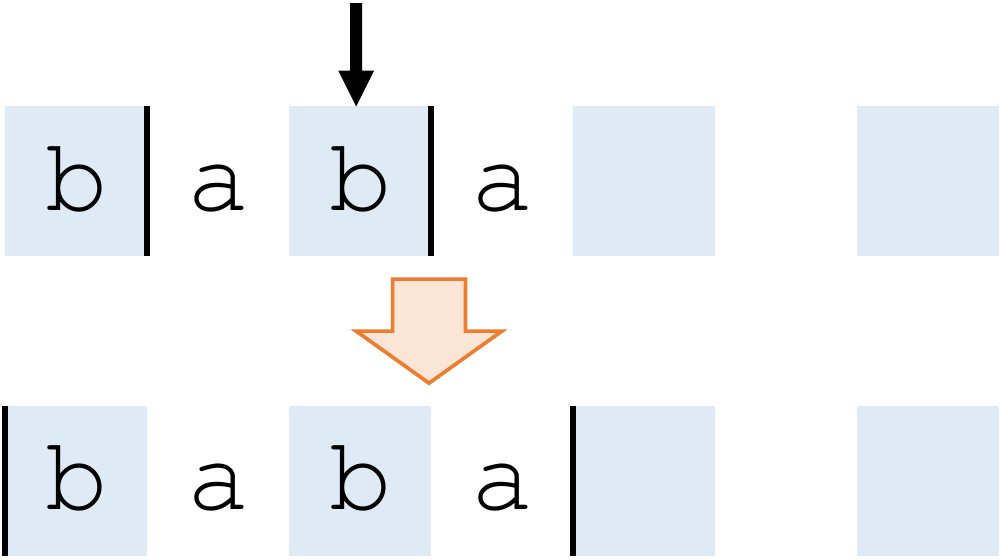}
      }
       &
      \parbox[c]{0.7\textwidth}{
        Case 4:
        Same as Case 1.
      }  \\\hline
      \parbox[c]{0.25\textwidth}{
        \includegraphics[width=0.23\textwidth]{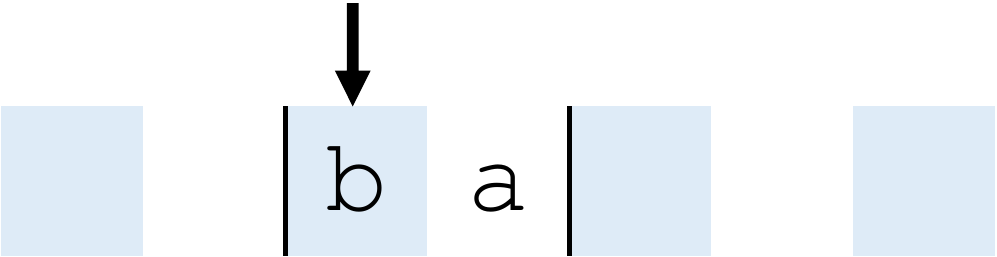}
      }
       &
      \parbox[c]{0.7\textwidth}{ \vspace{2pt}
      Case 5:
      There is nothing to do.
      We are done, since both
      $t_n[f_{B_n}^{j'-1}(i_\mathtt{b})]=\mathtt{b}$ and $t_n[f_{B_n}^{j'-1}(i_\mathtt{b})+1]=\mathtt{a}$
      are in a length-$2$ phrase starting at an even position, i.e., a relaxed ground phrase.
      }  \\\hline
    \end{tabular}
  }\caption{
    Terminal cases for eliminating a length-$3$ phrase $\mathtt{aba}$ that starts at an odd position $i_\mathtt{b}+1$ (see Section~\ref{secOddLength}).
    The shaded squares are even positions. The vertical bars denote phrase boundaries.
    The black arrow points to the position $f^{j'-1}_{B_n}(i_\mathtt{b})$, where $j'\geq 1$ is the smallest integer such that
    the parities of $f^{j'-1}_{B_n}(i_\mathtt{b})$ and $f^{j'}_{B_n}(i_\mathtt{b})$ differ.
    The first and second lines in Cases 1, 3, 4 show the phrase boundaries before and after the modification.
    The characters outside the phrase considered for each case can be inferred from being a partner of a phrase, and also from Lemma~\ref{lem:overlapfree}.
  }\label{fig:elim_length3}
\end{figure}

\subsubsection{Analysis of the Number of Phrases}
It is easy to see that Steps~\ref{step:elim_odd} and~\ref{step:inverse_morphism} do not increase the number of phrases.
Also, Step~\ref{step:elim_odd} does not decrease the number of length-$2$ phrases that start at even positions,
i.e., relaxed ground phrases, created in Step~\ref{step:elim_ground},
which will become ground phrases in $B_{n-1}$.
Thus, we focus on the analysis of Step~\ref{step:elim_ground}.

Examining each case of Fig.~\ref{fig:elim_ground}, we can see that while at the start
we eliminate a length-$1$ ground phrase and decrease the number of phrases,
Cases 1, 2, 3-1, and 4
introduce a new phrase, thus do not change the total number of phrases.
Also, notice that in Case 6, two ground phrases are eliminated, while the total number of phrases decreases only by one, since the second length-$1$ ground phrase is expanded.
Case 3-1 can occur in total at most twice, once for consecutive phrases of $\mathtt{ba}$ and
once for consecutive phrases of $\mathtt{ab}$.
Thus, we obtain the following inequality:
\begin{eqnarray}
  |B_{n-1}| \leq |B_n| - \lceil(\ground{B_n}-2)/2\rceil.\label{eqn:firstineq}
\end{eqnarray}

If $|B_{n-1}| \leq |B_n|-1$, then we can choose $i=1$ for Lemma~\ref{lemma:main} and are done.
Otherwise, $|B_{n-1}|=|B_n|$. This implies that $\ground{B_n} = 2$, and also that Case 3-1 was applied twice.
Thus, there exists at least 2 phrases of $\mathtt{ab}$ and $\mathtt{ba}$ each,
which are converted by $\mu^{-1}$ to ground phrases in $B_{n-1}$,
implying $\ground{B_{n-1}} \geq 4$.
Then, applying Equation~(\ref{eqn:firstineq}) for $n-2$, we have
\begin{eqnarray*}
  |B_{n-2}| &\leq& |B_{n-1}| - \lceil(\ground{B_{n-1}}-2)/2\rceil\\
  &\leq&|B_{n-1}| - 1 = |B_n| - 1.
\end{eqnarray*}

If $|B_{n-2}| \leq |B_n|-2$, then we can choose $i=2$ for Lemma~\ref{lemma:main}.
Otherwise, $|B_{n-2}|=|B_n|-1$.
This implies that $\ground{B_{n-1}} = 4$ and that Case 3-1 was applied twice, and Case 6 was applied once.
Therefore, we get $\ground{B_{n-2}}\geq 5$.
Finally, applying Equation~(\ref{eqn:firstineq}) for $n-3$, we have
\begin{eqnarray*}
  |B_{n-3}| &\leq& |B_{n-2}|  - \lceil(\ground{B_{n-2}}-2)/2\rceil\\
  &\leq& |B_{n-2}| - 2\\
  & =& |B_n| - 3.
\end{eqnarray*}
This proves Lemma~\ref{lemma:main}, and thus Theorem~\ref{thm:lowerbound}.

\section{Conclusion}
We have shown that for any $n\geq 2$, the size $b(t_n)$
of the smallest bidirectional scheme for the $n$-th Thue--Morse word $t_n$ is exactly $n+2$.
From the result that the smallest string attractor of $t_n$ is $4$ for any $n\geq 4$~\cite{DBLP:conf/spire/KutsukakeMNIBT20}
and that $|t_n| = 2^n$,
we have shown that Thue--Morse words are an example of a family of strings $\{S_n\}_{n \ge 1}$ in which each string~$S_n$ has
$b(S_n) = \Theta(\gamma(S_n) \log \frac{|S_n|}{\gamma(S_n)})$ as the size of its smallest bidirectional parsing,
where $\gamma(S_n)$ is the size of its smallest string attractor, and $|S_n| = 2^n$ is its length.
Note that we can generalize this to hold for any $\gamma\geq 4$:
Given a $\gamma \geq 4$, concatenate $k=\lfloor \gamma/4\rfloor$ copies of $t_n$,
each using distinct letters from a different binary alphabet.
Finally, we add $(\gamma\bmod 4)$ more distinct characters to make the smallest string attractor of the
resulting string exactly $\gamma$.
We thus can obtain a string of length $N=k\cdot 2^n+O(1)$ with $b =\Theta(kn)=\Theta(\gamma\log\frac{N}{\gamma})$.

Our result shows for the first time the separation between $\gamma$ and $b$,
i.e., there are string families such that $\gamma=o(b)$.
Whether this can be achieved by a family of binary strings is not yet known.
Although it is still open  whether $O(\gamma\log N)$ bits is enough to represent any string of length $N$,
it seems not possible by dictionary compression, i.e., copy/pasting within the string.
\bibliographystyle{splncs04}
\bibliography{refs}

\end{document}